\documentclass{llncs}
\usepackage{euscript}
\usepackage{amssymb}
\usepackage{pst-all}
\usepackage{stmaryrd}
\usepackage{latexsym}
\usepackage[all]{xy}
\usepackage{url}

\newcommand{\U}{\;\mathcal{U}\;}
\newcommand{\sema}[2]{\llbracket #1 \rrbracket_{#2}}
\newcommand{\sem}[1]{\llbracket #1 \rrbracket}
\newcommand{\bfup}[1]{\textup{\textbf{#1}}}

\newcommand{\<}{\sqsubseteq}

\newcommand{\lra}{\longrightarrow}

\newcommand{\ps}{\mathcal{P}}
\newcommand{\Rel}{\textup{\textbf{Rel}}}
\newcommand{\rel}[1]{\mathrm{Rel}(#1)}
\newcommand{\pred}[1]{\mathrm{Pred}(#1)}
\newcommand{\Sets}{\textup{\textbf{Sets}}}

\begin{document}

\title{Reflection and preservation of properties in coalgebraic
(bi)simulations\thanks{Research supported by the Spanish
projects DESAFIOS TIN2006-15660-C02-01, WEST TIN2006-15578-C02-01 and PROMESAS
S-0505/TIC/0407}}

\author{Ignacio F\'abregas \and Miguel Palomino \and David de
Frutos Escrig}
\institute{Departamento de Sistemas Inform\'aticos y
Computaci\'on,  UCM\\
\email{fabregas@fdi.ucm.es \quad \{miguelpt, defrutos\}@sip.ucm.es}}

\maketitle

\begin{abstract}
Our objective is to extend the standard results of preservation
and reflection of properties by bisimulations to the coalgebraic
setting, as well as to study under what conditions these results
hold for simulations. The notion of bisimulation is the classical
one, while for simulations we use that proposed by Hughes and
Jacobs. As for properties, we start by using a generalization of
linear temporal logic to arbitrary coalgebras suggested by Jacobs,
and then an extension by Kurtz which includes atomic propositions
too.
\end{abstract}

\section{Introduction}

To reason about computational systems it is customary to mathematically
formalize them by means of state-based structures such as labelled
transitions systems or Kripke structures. This is a fruitful approach since it
allows to study the properties of a system by relating it to some other,
possibly better-known system, by means of simulations and bisimulations
(see e.g., \cite{Milner89,LoiseauxEtAl95,KestenPnueli00,Clarke99}).

The range of structures used to formalize computational systems is
quite wide. In this context, coalgebras have emerged with a unifying aim \cite
{Rutten00}. A coalgebra is simply a function $c : X \longrightarrow FX$, where
$X$ is the set of states and $FX$ is some expression on $X$ (a functor) that
describes the possible outcomes of a transition from a given state.
Choosing different expressions for $F$ one can obtain coalgebras that
correspond to transition systems, Kripke structures, \dots

Coalgebras can also be related by means of (bi)simulations.
Our goal in this paper is to prove that, like their concrete
instantiations, (bi)simulations between arbitrary coalgebras preserve some
interesting properties. A first step in this direction consists in choosing an
appropriate notion for both bisimulation and simulation, as well as a logic in
which to express these properties.

Bisimulations were originally introduced by Aczel and Mendler
\cite{AczelMendler89}, who showed that the general definition
coincided with the standard ones when particularized; it is an established
notion. Simulations, on the other hand, were defined by Hughes and Jacobs
\cite{HughesJacobs04} and lack such canonicity. Their notion of simulation
depends on the use of orders that allow (perhaps too) much flexibility in what
it can be considered as a simulation; in order to show that simulations preserve
properties, we will have to impose certain restrictions on such orders.
As for the logic used for the properties, there is likewise no canonical
choice at the moment. Jacobs proposes a temporal logic (see \cite{Jacobs07})
that generalizes linear temporal logic (LTL), though without atomic
propositions; a clever insight of Pattinson \cite{Pattinson-thesis} provides us
with a way to endow Jacobs' logic with atomic propositions.

Since our original motivation was the generalization of the results
about simulations and preservation of LTL properties, we will focus on Jacobs'
logic and its extension with atomic propositions. Actually, modal logic seems to
be the right logic to express properties of coalgebras and several proposals
have been made in this direction, among them those in
\cite{Jacobs00,Kurz-thesis,Pattinson-thesis}, which are invariant under
behavioral equivalence. The reason for studying preservation/reflection of
properties by bisimulations here is twofold: on the one hand, some of the
operators in Jacobs' logic do not seem to fall under the framework of those
general proposals; on the other hand, some of the ideas and insights developed
for that study are needed when tackling simulations.
As far as we know, reflection of properties by simulations in coalgebras
has not been considered before in the literature.

\section{Preliminaries}

In this section we summarize definitions and concepts
from \cite{HughesJacobs04,JacobsRutten97,Jacobs07}, and introduce the
notation we are going to use.

Given a category $\mathbb{C}$ and an endofunctor $F$ in $\mathbb{C}$, an
$F$-coalgebra, or just a coalgebra, consists of an object $X\in\mathbb{C}$
together with a morphism $c:X\lra FX$. We often call $X$ the state space and
$c$ the transition or coalgebra structure.

\begin{example}
We show how two well-known structures can be seen as coalgebras:
\begin{itemize}
\item Labelled transition systems are coalgebras for the functor
$F=\ps(id)^A$, where $A$ is the set of labels.

\item Kripke structures are coalgebras for the functor
$F=\ps(AP)\times\ps(id)$, where $AP$ is a set of atomic propositions.
\end{itemize}
\end{example}

It is well-known that an arbitrary endofunctor $F$ on $\Sets$ can be lifted to
a functor in the category $\Rel$ of relations, that is, $\rel{F}:\Rel\lra\Rel$.
Given a relation $R\subseteq X\times Y$, its lifting is defined by

\[
\rel{F}(R) = \{ \langle u,v \rangle \in FX_1 \times FX_2 \mid
                \exists w\in F(R).\, F(r_1)(w) = u, F(r_2)(w) = v \}\; ,
\]

\noindent where $r_i:R\lra X_i$ are the projection morphisms.

A predicate $P$ of a coalgebra $c:X\lra FX$ is just a subset of the
state space. Also, a predicate $P\subseteq X$ can be lifted to a functor
structure using the relation lifting:

\begin{center}
$\pred{F}(P) = \coprod_{\pi_1}(\rel{F}(\coprod_{\delta}(P)))=
\coprod_{\pi_2}(\rel{F}(\coprod_{\delta}(P)))$,
\end{center}

\noindent where $\delta=\langle id,id\rangle$ and $\coprod_f(X)$ is the image
of $X$ under $f$, so $\coprod_{\delta_{x}}(P)=\{(x,x)\mid x\in
P\}$, $\coprod_{\pi_1}(R)=\{x_1\mid \exists x_2. x_1Rx_2\}$ is the domain of the
relation $R$, and $\coprod_{\pi_2}(R)=\{x_2\mid \exists x_1. x_1Rx_2\}$ is its
codomain.

The class of polynomial endofunctors is defined as the
least class of endofunctors on $\Sets$ such that it contains the
identity and constant functors, and is closed
under product, coproduct, constant exponentiation, powerset and
finite sequences. For polynomial endofunctors, $\rel{F}$ and $\pred{F}$ can be
defined by induction on the structure of $F$. For further details
on these definitions see \cite{Jacobs07}; we will introduce
some of those when needed. For example, for the cases
of labelled transition systems and Kripke structures we have:
\vspace{-.2cm}
\begin{eqnarray*}
\rel{\ps(id)^A}(R) & = \{(f,g)\mid \forall a\in A.\, (f(a),g(a))\in \{(U,V)
\mid
   & \forall u\in U.\, \exists v\in V.\, uRv \land \\
 & & \forall v\in V.\, \exists u\in U.\, uRv \}\}
\end{eqnarray*}
\vspace{-.9cm}
\begin{eqnarray*}
\pred{\ps(id)^A}(P) = \{f\mid\forall a\in A.\, f(a)\in \{U \mid \forall u\in
U.\, Pu \}\}
\end{eqnarray*}
\vspace{-.9cm}
\begin{eqnarray*}
\rel{\ps(AP)\times\ps(id)}(R) = & \{((u_1,u_2),(v_1,v_2)) & \mid (u_1=v_1.\,
u_1,v_1\in\ps(AP)) \land\\
& \phantom{\{((}(u_2,v_2)\in \{(U,V) & \mid \forall u\in U.\, \exists v\in V.\,
uRv \land \\
& & \phantom{\forall}\forall v\in V.\, \exists u\in U.\, uRv \}\}
\end{eqnarray*}
\vspace{-.9cm}
\begin{eqnarray*}
\pred{\ps(AP)\times\ps(id)}(P) = \{(u,v)\mid\ (u\subseteq\ps(AP)) \land (v\in
\{U \mid \forall u\in U.\, Pu )\}
\end{eqnarray*}

A bisimulation for coalgebras $c : X\lra FX$ and $d:Y \lra FY$ is a relation
$R\subseteq X\times Y$ which is ``closed under $c$ and $d$'':
\[
\textrm{if $(x,y) \in R$ then $(c(x), d(y)) \in \rel{F}(R)$}\,.
\]

\noindent In the same way, an invariant for a coalgebra $c : X\lra FX$ is a
predicate
$P\subseteq X$ such that it is ``closed under $c$'', that is, if $x \in
P$ then $c(x) \in \pred{F}(P)$.

We will use the definition of simulation introduced by
Hughes and Jacobs in \cite{HughesJacobs04} which uses an order $\<$ for
functors $F$ that makes the following diagram commute

\[
\xymatrix@R=4.0ex{
  & {\bfup{PreOrd}}\ar[d]^{\mathit{forget}} \\
 {\Sets}\ar[ur]^{\<} \ar[r]_{F} & \Sets
}
\]

\noindent Given an order $\<$ on $F$, a simulation for the
coalgebras $c: X\lra FX$ and $d: Y\lra FY$ is a relation $R\subseteq X\times Y$
such that
\[
\textrm{if $(x,y) \in R$ then $(c(x), d(y)) \in \rel{F}_\<(R)$}\,,
\]
where $\rel{F}_{\<}(R)$ is defined as
\[
\rel{F}_{\<}(R)=\{(u,v)\mid\exists w\in F(R).\; u\<
Fr_1(w)\wedge
Fr_2(w)\< v\} \,.
\]

To express properties we will use a generalization of LTL proposed by Jacobs
(see \cite{Jacobs07}) that applies to arbitrary coalgebras, whose formulas
are given by the following BNF expression:
\begin{displaymath}
\varphi= P\subseteq X\mid \neg \varphi \mid \varphi\vee \varphi
\mid \varphi\wedge \varphi \mid \varphi\Rightarrow \varphi \mid
\bigcirc \varphi \mid \Diamond \varphi \mid \Box \varphi \mid
\varphi \U \varphi
\end{displaymath}

\noindent $\bigcirc$ is the \textit{nexttime} operator and its semantics
(abusing notation) is defined as $\bigcirc P= c^{-1}(\pred{F}(P)) = \{x\in X\mid
c(x)\in \pred{F}(P)\}$; $\Box$ is the \textit{henceforth} operator defined as
$\Box P$ if exists an invariant for $c$, such that $Q\subseteq
P$ with $x\in Q$ or, equivalently by means of the greatest fixed point
$\nu$, $\Box P = \nu S.(P\wedge\bigcirc S)$; $\Diamond$ is the
\textit{eventually} operator defined as $\Diamond P=\neg\Box\neg P$; and $\U$ is
the \textit{until} operator defined as $P\U Q= \mu
S.(Q\vee(P\wedge\neg\bigcirc\neg S))$, where $\mu$ is the least fixed point.

We denote the set of states in $X$ that satisfies $\varphi$
as $\sema{\varphi}{X}$. That is, if $P\subseteq X$ is a predicate,
then $\sema{P}{X}=P$; if $\alpha\in\{\neg,\bigcirc,\Box,\Diamond\}$ then
$\sema{\alpha\varphi}{X}=\alpha\sema{\varphi}{X}$, and if
$\beta\in\{\wedge,\vee,\Rightarrow,\U\}$ then
$\sema{\varphi_1\beta\varphi_2}{X}=\sema{\varphi_1}{X}\beta\sema{\varphi_2}{X}
$. We will usually omit the reference to the set $X$ when it is
clear from the context. We say that an element $x$ satisfies a formula
$\varphi$, and we denote it by $c,x\models \varphi$, when $x\in
\sem{\varphi}$. Again, we will usually omit the reference to the coalgebra $c$.

\section{Reflection and preservation in bisimulations}

These definitions of reflection and preservation are slightly more involved
than for classical LTL because the logic proposed by Jacobs does not use atomic
propositions, but predicates (subsets of the set of states). Later, we will see
how atomic propositions can be introduced in the logic.

Given a predicate $P$ on $X$ and a binary relation $R\subseteq
X\times Y$, we will say that an element $y\in Y$ is in the direct
image of $P$, and we will denote it by $y\in RP$, if there
exists $x\in X$ with $x\in P$ and $x R y$. The inverse image of $R$ is just the
direct
image for the relation $R^{-1}$.

\begin{definition}
Given two formulas $\varphi$ on $X$ and $\psi$ on $Y$, built over
predicates $P_1,\dots P_n$ and $Q_1,\dots Q_n$, respectively, and
a binary relation $R\subseteq X\times Y$, we define the image of
$\varphi$ as a formula $\varphi^{*}$ on $Y$, obtained by substituting in
$\varphi$ $RP_i$ for $P_i$. Likewise, we build $\psi^{-1}$,
the inverse of $\psi$, substituting $R^{-1}Q_i$ for $Q_i$ in $\psi$.
\end{definition}

\begin{remark}\label{remark}
It is important to notice that $\varphi^*$ coincides with $\varphi^{-1}$ when
we consider $R^{-1}$ instead of $R$. Analogously, $\varphi^{-1}$ is
just $\varphi^{*}$ when we consider $R^{-1}$ instead of $R$.
\end{remark}

Now we can define when a relation preserves or reflects
properties.

\begin{definition}
Let $R\subseteq X\times Y$ be a binary relation and $a$ and $b$
elements such that $a R b$. We say that $R$ \emph{preserves}
the property $\varphi$ on $X$ if, whenever $a\models \varphi$, $b\models
\varphi^{*}$. We say that $R$ \emph{reflects} the property
$\varphi$ on $Y$ if $b\models \varphi$ implies $a\models
\varphi^{-1}$.
\end{definition}

Let us first state a couple of technical lemmas whose proofs appear in
\cite{FdFP07d}.

\begin{lemma}\label{lem-dir}
Let $F$ be a polynomial functor, $R\subseteq X\times Y$ a
bisimulation between coalgebras $c:X \lra FX$ and $d:Y\lra FY$,
$P\subseteq Y$, $Q\subseteq X$  and $xRy$. If $d(y)\in\pred{F}(P)$,
then $c(x)\in\pred{F}(R^{-1}P)$; and if $c(x)\in\pred{F}(Q)$, then
$d(y)\in\pred{F}(RQ)$.
\end{lemma}

Another auxiliary lemma we need to prove the main result of this section is the
following:

\begin{lemma}\label{lem-invariant}
The direct and inverse images of an invariant are also invariants.
\end{lemma}

\begin{proof}
Let $R$ be a bisimulation between $c:X\lra FX$ and $d:Y\lra FY$.
Let us suppose that $P\subseteq X$ is an invariant and let us prove
that so is $R P$; that is, for all $y\in R P$ it must be the case
that $d(y)\in\pred{F}(R P)$. If $y\in R P$, then
there exists $x\in P$ such that $xR y$. Since $P$ is an invariant, we
also have $c(x)\in\pred{F}(P)$ and by Lemma \ref{lem-dir} we get
$d(y)\in\pred{F}(R P)$.

On the other hand, since $R^{-1}$ is also a bisimulation, the
inverse image of an invariant is an invariant too. \qed
\end{proof}

At this point it is interesting to recall that our
objective is to prove that bisimulations preserve and reflect
properties of a temporal logic, that is, if we have $xR y$ and
$y\models\varphi$ then we must also have $x\models\varphi^{-1}$;
and, analogously, if $x\models\varphi$ then $y\models\varphi^{*}$. We will
show this result for all temporal operators except for the negation; it is
well-known that negation is reflected and preserved by standard
bisimulations, but not here because of the lack of atomic propositions in the
coalgebraic temporal logic.

To prove the result for the rest of temporal operators, we will see
that if $y\in\sem{\varphi}$ then we also have $x\in
R^{-1}\sem{\varphi}$ and, analogously, if $x\in\sem{\varphi}$ then
$y\in R\sem{\varphi}$. Ideally, we would like to have both $R^{-1}\sem{\varphi}=
\sem{\varphi^{-1}}$ and $R\sem{\varphi}=\sem{\varphi^{*}}$ but, in
general, only the inclusion $\subseteq$ is true. Fortunately this
is enough to prove our propositions, since the temporal operators are
all monotonic except for the negation. In fact, here is where the
problem with negation appears.

\begin{lemma}[\cite{FdFP07d}]\label{l-formula inversa}
Let $R$ be a bisimulation between coalgebras $c:X\lra FX$ and
$d:Y\lra FY$. For all temporal formulas $\varphi$ and $\psi$ which do not
contain the negation operator, it follows that
\[R^{-1}\sema{\varphi}{Y}\subseteq\sema{\varphi^{-1}}{X}\quad
\textrm{and}\quad R\sema{\psi}{X}\subseteq\sema{\psi^{*}}{Y}\,.\]
\end{lemma}

Finally we can show that bisimulations reflect and preserve
properties given by any temporal operator except for the negation.

\begin{proposition}
Let $\psi$ be a formula over a set $Y$ which does not use the
\emph{negation} operator and let $R$ be a bisimulation between
coalgebras $c:X\lra FX$ and $d:Y\lra FY$. Then the property $\psi$
is reflected by $R$.
\end{proposition}

\begin{proof}
The result is proved by structural induction over the formula
$\psi$ using the first half of Lemmas \ref{lem-dir} and \ref{l-formula
inversa}, and Lemma \ref{lem-invariant}. See \cite{FdFP07d} for further
details. \qed
\end{proof}

Preservation of properties is a consequence of the
reflection of properties together with the fact that if $R$ is a bisimulation
then $R^{-1}$ is also a bisimulation. We have thus proved the
following theorem.

\begin{theorem}\label{t-bis}
Let $\psi$ and $\varphi$ be formulas over sets $Y$ and $X$,
respectively, which do not use the \emph{negation} operator and
let $R$ be a bisimulation between coalgebras $c:X\lra FX$ and
$d:Y\lra FY$. Then $\psi$ is reflected by $R$ and $\varphi$ is
preserved by $R$.
\end{theorem}

\section{Reflection and preservation in simulations}

In \cite{Clarke99,Palomino-Thesis} it is proved not only that bisimulations
reflect and preserve properties but also that simulations reflect
them: it turns out that this result does not generalize straightforwardly to
the coalgebraic setting.

The main problem that we have found concerning this is that the
coalgebraic definition of simulation uses an arbitrary functorial
order $\<$, and in general reflection of properties will not hold for all
orders.

Let us show a counterexample that will convince us that
simulations may not reflect properties without restricting the
orders. Let us take
$F=\ps(id)$, $X=\{x_1,x_2\}$, $Y=\{y_1,y_2\}$ and the coalgebras
$c$ and $d$ defined as $c(x_1)=\{x_1,x_2\}$, $c(x_2)=\{x_2\}$,
$d(y_1)=y_2$ and $d(y_2)=y_2$. We define $u\< v$ whenever
$v\subseteq u$ and consider the formula $\varphi=\bigcirc P$,
where $P=\{y_2\}$, and the simulation $R=\{(x_1,y_2)\}$. It is
immediate to check that $R$ is a simulation and
$y_2\in\sem{\varphi}$, but $x_1\notin\sem{\varphi^{-1}}$.
\begin{itemize}
\item $y_2\in\sem{\varphi}$. Indeed, since $d(y_2)=y_2$ then
$y_2\in\sem{\varphi}=\bigcirc P$ is equivalent to $y_2\in
P=\{y_2\}$, which is trivially true.

\item $x_1\notin\sem{\varphi^{-1}}$. By definition,
$\varphi^{-1}=\bigcirc R^{-1}P=\bigcirc\{x_1\}$. Since
$c(x_1)=\{x_1,x_2\}$, it is enough to see that $x_2\notin\{x_1\}$,
which is also true.
\end{itemize}

As a consequence, we will need to restrict the functorial orders that
are involved in the definition of simulation. In a first approach we will impose
an extra requirement that the order must fulfill, and later we will not only
restrict the orders but also the functors that are involved.

\subsection{Restricting the orders}\label{sec-orders}

The idea is that we are going to require an extra property for
each pair of elements which are related by the order. In particular,
we are particularly interested in the following property (which is
defined in \cite{HughesJacobs04}):

\begin{definition}\label{down-closed}
Given a functor $F:\Sets\lra \Sets$, we say that an order $\<$
associated to it is ``down-closed'' whenever $a\< b$, with $a,b\in
FX$, implies that
\begin{displaymath}
b\in\pred{F}(P)\;\Longrightarrow\; a\in\pred{F}(P),\quad
\textrm{for all predicates $P\subseteq X$}\,.
\end{displaymath}
\end{definition}

We can show some examples of down-closed orders:

\begin{example}
\begin{enumerate}
\item Kripke structures are defined by the functor
$F=\ps(AP)\times\ps(id)$, so a down-closed order must fulfill that
if $(u,v)\< (u',v')$, then $(u',v')\in\pred{F}(P)$ implies
$(u,v)\in\pred{F}(P)$; that is, by definition of
$\pred{\ps(AP)\times\ps(id)}$, $u,u'\subseteq\ps(AP)$ and, if
$v'\in\pred{\ps(id)}(P)=\{U\mid \forall u\in U.\, u\in P\}$ then
$v\in\pred{\ps(id)}(P)$. In other words, for all $b\in v$ and $b'\in
v'$, if $b'\in P$ then $b\in P$. Therefore, what is needed
in this case is that the set of successors $v$ of the smaller pair is
contained in the set of successors $v'$ of the bigger pair, that
is, if $(u,v)\<(u',v')$ then $v\subseteq v'$.

\item Labelled transition systems are defined by the functor
$F=\ps(id)^{A}$, so the order must fulfill the following: if $u\<
v$ then $\forall a\in A.\; u(a)\subseteq u'(a)$.
\end{enumerate}
\end{example}

Those examples show that there are not many down-closed orders,
but it does not seem clear how to further extend this class in
such a way that we could still prove the reflection of properties
by simulations. Unfortunately, even under this restriction we can
only prove reflection (or preservation) of formulas that only use
the operators $\vee$, $\wedge$, $\bigcirc$ and $\Box$.

To convince us of this fact, we present a counterexample with
operator $\Diamond$. Let $X=\{x_1,x_2\}$, $Y=\{y_1,y_2\}$ and the functor
$F=\ps(id)$. We consider the following down-closed order: $u\< v$ if
$u\subseteq v$. We also define the coalgebras $c:X\lra
FX$ and $d:Y\lra FY$ as $c(x_1)=\{x_1\}$, $c(x_2)=\{x_2\}$,
$d(y_1)=\{y_1,y_2\}$ and $d(y_2)=\{y_2\}$. Obviously
$R=\{(x_1,y_1)\}$ is a simulation since $c(x_1)=\{x_1\}\< \{x_1\}$
and $\{y_1\}\< \{y_1,y_2\}=d(y_1)$ and, also,
$\{x_1\}\rel{F}(R)\{y_1\}$. We have $y_1\in\Diamond\{y_2\}$, since we
can reach $y_2$ from $y_1$, but $x_1\notin\Diamond
R^{-1}\{y_2\}=\Diamond\emptyset$. Indeed,
$x_1\notin\Diamond\emptyset$ is equivalent to
$x_1\in\Box\neg\emptyset$ and this is true since $\{x_1\}$ is an
invariant such that $x_1\in\{x_1\}$, with
$\{x_1\}\subseteq\neg\emptyset$.

In order to prove reflection of properties that only use the
operators $\vee$, $\wedge$, $\bigcirc$ and $\Box$, we will need a
previous elementary result involving binary relations.

\begin{proposition}\label{p-alzamientos1}
Let $R\subseteq X\times Y$ be a binary relation and $P\subseteq Y$ a
predicate. Let us suppose that $u\rel{F}(R)v$; then, if
$v\in\pred{F}(P)$ it is also true that $u\in\pred{F}(R^{-1}P)$.
\end{proposition}

\begin{proof}
Once again the proof will proceed by structural induction on the
functor $F$. See \cite{FdFP07d} for further details. \qed
\end{proof}

We will also need a subtle adaptation of Lemmas
\ref{lem-invariant} and \ref{l-formula inversa} from the framework
of bisimulations to the framework of simulations. In particular,
we can adapt Lemma \ref{lem-invariant} to prove that if $Q$ is an
invariant and $R$ a simulation, $R^{-1}Q$ is still an invariant,
whereas the first half of Lemma \ref{l-formula inversa} will also be true in the
framework of simulations for formulas that only use the operators
$\vee$, $\wedge$, $\bigcirc$ and $\Box$.

\begin{lemma}\label{l-invariantes y sim}
Let $R$ be a simulation between coalgebras $c:X\lra FX$ and
$d:Y\lra FY$, with a down-closed order, and let $Q\subseteq Y$
be an invariant. Then $R^{-1}Q$ is also an invariant.
\end{lemma}

\begin{proof}
We are going to show that for all $x\in R^{-1}Q$ we have
$c(x)\in\pred{F}(R^{-1}Q)$. Let us take an arbitrary $x\in
R^{-1}Q$; then, by definition there exists $y\in Q$ such that
$xRy$ and, since $Q$ is an invariant, $d(y)\in\pred{F}(Q)$. On the
other hand, since $R$ is a simulation, $c(x)\< u\rel{F}(R)v\<
d(y)$. Henceforth, since we are working with a down-closed order
and $d(y)\in\pred{F}(Q)$, then $v\in\pred{F}(Q)$. Also, by
Proposition \ref{p-alzamientos1} we have $u\in\pred{F}(R^{-1}Q)$
and, using again that the order is down-closed, it follows that
$c(x)\in\pred{F}(R^{-1}Q)$. \qed
\end{proof}

\begin{lemma}[\cite{FdFP07d}]\label{l-formula inversa sim}
Let $R$ be a simulation between coalgebras $c:X\lra FX$ and
$d:Y\lra FY$, with a down-closed order. If $\varphi$ is a
temporal formula constructed only with operators $\vee$, $\wedge$,
$\bigcirc$ and $\Box$, then
\[R^{-1}\sema{\varphi}{Y}\subseteq\sema{\varphi^{-1}}{X}\,.\]
\end{lemma}

Now we can state the corresponding theorem:

\begin{theorem}[\cite{FdFP07d}]\label{t-sim1}
Let $R$ be a simulation between coalgebras $c:X\lra FX$ and
$d:Y\lra FY$ with a down-closed order. If $\varphi$ is a
temporal formula constructed only with operators $\vee$, $\wedge$,
$\bigcirc$ and $\Box$, then the property $\varphi$ is reflected by
the simulation.
\end{theorem}

Instead of considering down-closed orders, we
could have imposed the converse implication, that is, those
orders that satisfy that if $a\in\pred{F}(P)$ then
$b\in\pred{F}(P)$.

\begin{definition}
Given a functor $F:\Sets\lra \Sets$ we say that an order
$\<$ is \emph{up-closed} if whenever $a\< b$ then
\begin{displaymath}
a\in\pred{F}(P)\;\Longrightarrow\; b\in\pred{F}(P),\quad
\textrm{for all predicates $P$}\,.
\end{displaymath}
\end{definition}

Obviously up-closed is symmetrical to down-closed, that is, it is
equivalent to taking $\<^{op}$ instead of $\<$ in Definition
\ref{down-closed}. So, for example, in the case of Kripke
structures an up-closed order would satisfy
$(u,v)\<(u',v')$ if $v'\subseteq v$.

The interesting thing about up-closed orders is that they allow us
to prove \emph{preservation} of properties; again, this result
will hold only for formulas constructed with the operators $\vee$,
$\wedge$, $\bigcirc$ and $\Box$. We need the following auxiliary
result whose proof is analogous to the case of down-closed
orders.
Since if $R$ is a simulation for the order $\<$, then $R^{-1}$ is a
simulation for the oposite order $\<^{op}$, we can apply Theorem
\ref{t-sim1} to get the following (see \cite{FdFP07d} for more details):

\begin{theorem}
Let $R$ be a simulation between coalgebras $c:X\lra FX$ and
$d:Y\lra FY$ carrying an up-closed order. If $\varphi$ is
a temporal formula constructed only with the operators $\vee$,
$\wedge$, $\bigcirc$ and $\Box$, then $R$ preserves the property
$\varphi$.
\end{theorem}

\subsection{Restricting the class of functors}\label{sec-functors}

As we have just seen, it is not enough to restrict ourselves to
down-closed (or up-closed) orders to get a valid result for all
properties. What we want is a necessary and sufficient condition
over functorial orders that implies reflection (or preservation)
of properties by simulations. So far we have not found such a
condition, but we have a sufficient one for simulations to reflect
properties (and, in fact, also so that they preserve properties).

Recalling the structure of lemmas and propositions used to prove
reflection and preservation of properties by bisimulations, we
notice that the key ingredient was Lemma \ref{lem-dir}. With this lemma we
were able to prove directly preservation of invariants (Lemma
\ref{lem-invariant}) and the relation between $R^{-1}$ (respectively $R$) of a
formula and the inverse of a formula (respectively direct image of a formula).
Also, Lemma \ref{lem-dir} was essential to prove directly reflection and
preservation of formulas built with the \emph{nexttime} operator and the rest of
temporal operators.

In the previous section the problem we faced was that either the second half of 
Lemma \ref{lem-dir} (for down-closed orders) or the first half of Lemma
\ref{lem-dir} (for up-closed orders) held, but not both
simultaneously. As a consequence, the results for the operators
\emph{eventually} and \emph{until} did not hold. So, if we were
capable of finding a subclass of functors and orders such that they
fulfill results analogous to Lemma \ref{lem-dir} then, translating those proofs,
we would get reflection and preservation of arbitrary properties.

We are going to define a subclass of functors and orders in the
way that Hughes and Jacobs did in \cite{HughesJacobs04} for the
subclass \textbf{Poly}.

\begin{definition}\label{d-ordenes}
The class \bfup{Order} is the least class of functors closed under
the following operations:
\begin{enumerate}
\item For every preorder $(A,\leq)$, the constant functor
$X\mapsto A$ with the order given by $\<_{X}=\leq_{A}$.

\item The identity functor with equality order.

\item Given two polynomial functors $F_1$ and $F_2$ with orders
$\<^1$ and $\<^2$, the product functor
$F_1\times F_2$ with order $\<_{X}$ given by
\begin{displaymath}
(u,v)\<_{X}(u',v')\quad \textrm{if}\quad u\<^1
u'\quad \textrm{and}\quad v\<^2 v'\, .
\end{displaymath}

\item Given the polynomial functor $F$ with order $\<^F$
and the set $A$, the functor $F^A$ with order
$\<_{X}$ given by
\begin{displaymath}
u\<_{X}v\; \textrm{if}\;\; u(a)\<^F v(a)\; \textrm{for all $a\in
A$.}
\end{displaymath}

\item Given two polynomial functors $F_1$ and $F_2$ with orders
$\<^1$ and $\<^2$, the coproduct functor $F_1 + F_2$ with
order $\<_{X}$ given by
\begin{displaymath}
\begin{array}{lcl}
u\<_{X}v & \textrm{if} & u=\kappa_1(u_0)\;
\textrm{and}\; v=\kappa_1(v_0)\;\textrm{with}\; u_0\<^1 v_0\\
& \textrm{or} & u=\kappa_2(u_0)\; \textrm{and}\;
v=\kappa_2(v_0)\;\textrm{with}\; u_0\<^2 v_0\; .
\end{array}
\end{displaymath}

\item Given the polynomial functor $F$ with order $\<^F$, the powerset
functor $\ps(F)$ with order
$\<_{X}$ given by
\begin{displaymath}
\begin{array}{lcl}
u\<_{X}v & \textrm{if} & \forall a\in u\; \exists b\in
v\quad\textrm{such that}\quad a\<^F b\\
& \textrm{and also} & \forall b\in v\; \exists a\in
u\quad\textrm{such that}\quad a\<^F b\, .
\end{array}
\end{displaymath}
\end{enumerate}
\end{definition}

For example the usual order for Kripke structures is not in the
class \textbf{Order}. Besides, in the definition of \textbf{Poly}
in \cite{HughesJacobs04} the authors did not consider the powerset
functor but we do, although we are not using the \emph{usual}
order for this functor.

At first, to obtain that simulations not only reflect but also
preserve properties may seem a little surprising. If we think
about the elements in the subclass \bfup{Order} we notice that we
have restricted the orders to equality-like orders, that is,
almost all possible orders in \bfup{Order} are the equality.
However, since the class \textbf{Order} is very similar to the
class \textbf{Poly}, it has the same good properties shown in
\cite{HughesJacobs04} (like the stablility of the orders and
functors). 

\begin{example}
\begin{enumerate}
\item If we consider the functor $\ps(id)$, then the order $\<$
defined in Definition \ref{d-ordenes} says that $u\< v$ if and
only if for each $a\in u$ there exists $b\in v$ such that $a=b$,
and if for each $b\in v$ there exists $a\in u$ such that $a=b$.
This means that $\<$ is the identity relation. As an immediate
consequence for transition systems the only possible
\textbf{Order} simulations are bisimulations.

\item If we consider the functor $A\times id$ where $A$ has a
preorder $\leq_A$ different from the identity, the order $\<$
from Definition \ref{d-ordenes} is the following:
$(u,v)\<(u',v')$ iff $v=v'$ and $u\leq_{A} u'$. So, if $\leq_A$ is
not the identity, neither is $\<$. For example, let us take
$X=\{x_1,x_2,x_3\}$, $Y=\{y_1,y_2\}$, $AP=\{p_1,p_2,p_3\}$ and
consider the functor $F=\ps(id)\times\ps(AP)$ and the coalgebras
$c:X\lra FX$ and $d:Y\lra FY$ defined by
$c(x_1)=(\{x_2,x_3\},\{p_1\})$, $c(x_2)=(\{x_3\},\{p_2\})$,
$c(x_3)=(\{x_2\},\{p_3\})$, $d(y_1)=(\{y_2\},\{p_2\})$ and
$d(y_2)=(\{y_2\},\{p_1\})$. Obviously there is no bisimulation
between $x_1$ and $y_1$ since this atomic propositions are not
the same, but taking the order $\<$ defined as
$(u,v)\<(u',v')$ iff $u=u'$ (that is, taking as the preorder $\leq_{AP}$ the
total relation) we have that there exists a simulation $R$ in \textbf{Order}
between $x_1$ and $y_1$. 
\end{enumerate}
\end{example}

\begin{lemma}[\cite{FdFP07d}]\label{lem-inv simulaciones}
Let $R\subseteq X\times Y$ be a simulation between coalgebras
$c:X\lra FX$ and $d:Y\lra FY$, such that the functor $F$ is in the
class \bfup{Order}. Let us also suppose that $P\subseteq Y$ and $x
R y$; then, if $d(y)\in\pred{F}(P)$ we have
$c(x)\in\pred{F}(R^{-1}P)$.
\end{lemma}

In a similar way we have the corresponding lemma involving direct
predicates.

\begin{lemma}\label{lem-dir simulaciones}
Let $R\subseteq X\times Y$ be a simulation between coalgebras
$c:X\lra FX$ and $d:Y\lra FY$, such that the functor $F$ is in
\bfup{Order}. Let us suppose also that $P\subseteq X$ and $x R y$.
Then, if $c(x)\in\pred{F}(P)$, $d(y)\in\pred{F}(R P)$.
\end{lemma}

Now we can conclude that under these hypothesis simulations
reflect and preserve properties, simultaneously! This fact is a
straightforward result from Lemmas \ref{lem-inv simulaciones} and
\ref{lem-dir simulaciones}.

\begin{theorem}
Let $R$ be a simulation between coalgebras $c:X\lra FX$ and
$d:Y\lra FY$, with $F$ a polynomial functor in the class
\bfup{Order}. Then, the simulation $R$ reflects and preserves
properties.
\end{theorem}

\section{Including atomic propositions}

A consequence of the fact that the logic proposed by Jacobs
does not introduce atomic propositions was the need of giving
non-standard definitions of reflection and preservation of
properties. Kurz, in his work \cite{Kurz-thesis} includes atomic
propositions in a temporal logic for coalgebras by means of
natural transformations.

\begin{definition}\label{d-sintaxis ext}
Given a set $AP$ of atomic propositions, the formulas of the
temporal logic associated to a coalgebra $c:X\lra FX$ are given
by the BNF expression:
\begin{displaymath}
\varphi= p\mid \neg \varphi \mid \varphi\vee \varphi \mid
\varphi\wedge \varphi \mid \varphi\Rightarrow \varphi \mid
\bigcirc \varphi \mid \Diamond \varphi \mid \Box \varphi \mid
\varphi \U \varphi
\end{displaymath}
where $p\in AP$ is an atomic proposition.
\end{definition}

Kurz also defines when a state $x$ satisfies an atomic proposition
$p$, that is, he defines the semantics of an atomic proposition.

\begin{definition}\label{d-at}
Let $F:\Sets\lra\Sets$ be a functor and $AP$ a set of atomic
propositions. Let $\nu: F\Rightarrow\ps(AP)$ be a natural
transformation and $c:X\lra FX$ a coalgebra. We say that $x$
satisfies an atomic proposition $p\in AP$, and denote it $x\models
p$, when $p\in(\nu_{X}\circ c)(x)$. This way $\sem{p}=\{x\mid p\in(\nu_{X}\circ
c)(x)\}$.
\end{definition}

Notice that in fact this defines not only a semantics but a family
of possible semantics that depends on the natural transformation.
For example, we can define a natural transformation for the
functor for Kripke structures in this way:
\begin{displaymath}
\begin{array}{lclc}
\nu_{X}: & \mathcal{P}(AP)\times\mathcal{P}(X) & \longrightarrow & \mathcal{P}(AP)\\
  & (P,Q) & \longmapsto & P
\end{array}
\end{displaymath}

\noindent With $\nu_{X}$ we have characterized the standard
semantics of LTL for Kripke structures. Analogously, we
could define the following interpretation:
$\nu'_{X}(P,Q)=\ps(AP)\setminus P$.

Introducing in our temporal logic the semantics of the atomic
propositions, we can prove the following theorem involving
bisimulations:

\begin{theorem}\label{p-bisimulaciones y at}
Let $R$ be a bisimulation between coalgebras $c:X\lra FX$ and
$d:Y\lra FY$. Let $\varphi$ be a temporal formula; then, the
following is true for all $x\in X$ and $y\in Y$ such that $x R y$:
\begin{displaymath}
x\in\llbracket \varphi \rrbracket_{X}\quad
\Longleftrightarrow\quad y\in\llbracket \varphi \rrbracket_{Y}\,.
\end{displaymath}
\end{theorem}

Here we have captured in the same theorem the classical ideas of
reflection and preservation of properties: we have some property
in the lefthand side of a bisimulation if and only if we have the
property in its righthand side. In this case the theorem is true
also for the negation operator thanks to the atomic propositions.
Intuitively, this is because now we have an ``if and only if''
theorem, whereas in Theorem \ref{t-bis} we needed to reason
separately for each implication using monotonicity, and negation
lacks it. Also notice that even though we could think that in Theorem
\ref{t-bis} our predicates played the role of atomic
propositions, there are some essential differences: first, predicates are not
independent of each other, unlike atomic propositions, and secondly, while
atomic propositions stay the same predicates vary with each set of states.

\begin{proof}
Once again the proof will proceed by structural induction on the
formula $\varphi$. We only show some of the cases (the complete proof can be
found in \cite{FdFP07d}).
\begin{enumerate}
\item Let $\varphi = p$ where $p$ is an arbitrary atomic proposition. This way
we have the following diagram, for $\nu$ an arbitrary natural
trasformation:

\[
\xymatrix@R=5.0ex{
 {X}\ar[d]_{c}       & {R} \ar[d]_{[c,d]}\ar[l]_{\pi_1}\ar[r]^{\pi_2}        &
Y\ar[d]_{d} \\
 {FX}\ar[d]_{\nu_X} & {FR} \ar[l]_{F\pi_1}\ar[r]^{F\pi_2}\ar[d]_{\nu_R} &
{FY}\ar[d]_{\nu_Y} \\
 {\ps(AP)}               & {\ps(AP)} \ar[l]_{id}\ar[r]^{id}               &
\ps(AP)
 }
\]

This diagram is commutative. Indeed, since $R$ is a bisimulation
the upper side commutes, while the lower side commutes because
$\nu$ is a natural transformation.

So, $x\in\sema{\varphi}{X}$ means by definition that $p\in(\nu_{X}\circ
c)(x)$. Since the diagram commutes then $p\in(\nu_{R}\circ [c,d])(x,y)\;
\Leftrightarrow\; p\in(\nu_{Y}\circ d)(y)$, that is, $y\in\sema{\varphi}{Y}$.

\item Let us suppose $\varphi=\neg\varphi_{0}$. In this case we must show
that $x\in\neg\sema{\varphi_0}{X}$ if and only if
$y\in\neg\sema{\varphi_0}{Y}$, that is, we must see
that $x\notin\sema{\varphi_0}{X}$ if and only if
$y\notin\sema{\varphi_0}{Y}$. By induction hypothesis we
have $x\in\sema{\varphi_0}{X}$ if and only if $y\in\sema{\varphi_0}{Y}$.

\item Let us suppose now that $\varphi=\bigcirc\varphi_{0}$. We
must prove that $x\in\bigcirc\sema{\varphi_{0}}{X}$ is equivalent
to $y\in\bigcirc\sema{\varphi_{0}}{Y}$, that is,
$c(x)\in\pred{F}(\sema{\varphi_{0}}{X})$ is equivalent to
$d(y)\in\pred{F}(\sema{\varphi_{0}}{Y})$. The latter will be
proved by structural induction on the functor $F$. As an example we show the
case of $F=G^A$. Let us prove only one implication since the
other
one is almost identical. We have
\begin{displaymath}
\pred{F}(\sema{\varphi_{0}}{X})=\{f\mid\forall a\in A.\; f(a)\in
\pred{G}(\sema{\varphi_{0}}{X})\}\,.
\end{displaymath}
Once again, as we have shown in other proofs, we define for each
$a\in A$ and each $F$-coalgebra $c:X\lra F(X)$ a $G$-coalgebra,
$c^{a} : X\lra G(X)$ where for each $x\in X$ we have $c^{a}(x) =
c(x)(a)$. In this way, we have $xR y$ and $c^{a}(x) =
c(x)(a)\in\pred{G}(\sema{\varphi_{0}}{X})$. By induction
hypothesis we have that
$d^{a}(y)\in\pred{G}(\sema{\varphi_{0}}{Y})$. Since this is a
valid argument for all $a\in A$, we obtain
$d(y)\in\pred{F}(\sema{\varphi_{0}}{Y})$.

\item $\varphi=\Box\varphi_{0}$. Assuming that
$x\in\sema{\varphi}{X}$ we get that there exists
\begin{displaymath}
Q\subseteq X\; \textrm{an invariant for $c$ with}\;
Q\subseteq\sema{\varphi_{0}}{X}\;\textrm{and}\; x\in Q.
\end{displaymath}
Now, $R Q$ is a invariant for $d$ and, also, such that
$RQ\subseteq\sema{\varphi_{0}}{Y}$ with $y\in RQ$. Indeed, if
$x\in Q$ then $y\in RQ$ and if $b\in RQ$ there must exists some
$a\in Q\subseteq\sema{\varphi_{0}}{X}$ such that $aRb$. So, by
induction hypothesis we get that $b\in\sema{\varphi_{0}}{Y}$

On the other hand, if $y\in\sema{\varphi}{Y}$ there must exists
some invariant $T$ on $Y$, such that
$T\subseteq\sema{\varphi_{0}}{Y}$ with $y\in T$, hence for proving
$x\in\sema{\varphi}{X}$ it is enough to consider the invariant
$R^{-1}T$. \qed
\end{enumerate}
\end{proof}

To obtain a similar result for simulations, we will need again to
restrict the class of functors and orders as we did in Sections
\ref{sec-orders} and \ref{sec-functors}. In particular we are
interested in the following antimonotonicity property: if $u\< u'$ then
$\nu(u')\subseteq \nu(u)$.

\begin{definition}
Let $F:\Sets\lra\Sets$ be a functor, $AP$ a set of atomic
propositions and $\nu: F\Rightarrow\ps(AP)$ a natural
transformation. We say that $\<$ is a down-natural $\nu$-order if,
whenever $u\< u'$ then $\nu(u')\subseteq\nu(u)$.
\end{definition}

Obviously this definition depends on the natural transformation
that we consider in each case. For example, for Kripke structures
we have the following natural transformation:
$\nu_{X}((A_X,B_X))=A_X\subseteq AP$. To
obtain a down-natural $\nu$-order the following must hold: $(u,v)\< (u',v')$
then $\nu((u',v'))\subseteq\nu((u,v))$, that is, it will be enough to require
$(u,v)\<(u',v')$ iff $u'\subseteq u$.

This way, if we combine the down-closed and the down-natural
orders we get:
\begin{displaymath}
\textrm{If}\quad (u,v)\<(u',v')\quad \textrm{then}\quad
u'\subseteq u\; \textrm{and}\; v\subseteq v'\,.
\end{displaymath}

This characterization is not as restrictive as one could think.
Indeed, if we recall the definition of functorial order we had:
\[
\xymatrix@R=4.0ex{
  & {\bfup{PreOrd}}\ar[d]^{\mathit{forget}} \\
 {\Sets}\ar[ur]^{\<} \ar[r]_{F} & \Sets
}
\]

\noindent This diagram means that the functor $F$ and the order
$\<$ almost have the same structure and indeed, we could use a
natural transformation between $\<$ and $\ps(AP)$ in Definition
\ref{d-at} instead of a natural transformation between $F$ and
$\ps(AP)$, that is, $\nu:\<\Rightarrow\ps(AP)$. Considering $\nu$
in this way, an immediate consequence is that if we take as order in
$\ps(AP)$ the relation $\supseteq$ (as is done in \cite{Palomino-Thesis}), then
$u\< v$
implies $\nu(u)\,\<\nu(v)$.

We can tackle the proof of reflection of properties (with atomic
propositions) by simulations as we did in Section
\ref{sec-orders}, imposing to the order not only to be
down-natural but also down-closed. But, if we do that we will find
the same difficulties we faced in Section \ref{sec-orders}
(that is, we would not be able to prove reflection of formulas
built with the operators \emph{until} and \emph{eventually}).
Therefore, we must restrict the class of functors and orders, as we
did with the class \textbf{Order} in Section \ref{sec-functors},
but imposing also that the orders must be down-natural.

\begin{definition}
The class \bfup{Down-Natural $\nu$-Order} is the subclass of
\bfup{Order} where all orders are down-natural.
\end{definition}

Notice that we are defining a different class for each natural
transformation $\nu$. Under this condition we state the
corresponding theorem involving simulations and the reflection of
properties (with atomic propositions); for the proof see
\cite{FdFP07d}.

\begin{theorem}
Let $R$ be a simulation between coalgebras $c:X\lra FX$ and
$d:Y\lra FY$ on the same polynomial functor $F$ from $\Sets$ to
$\Sets$ belonging to the class \bfup{Down-Natural $\nu$-Order} and let
$\varphi$ be a temporal formula. Then, for each $x\in X$ and $y\in
Y$ such that $x R y$:
\begin{displaymath}
y\in\sema{\varphi}{Y}\quad \Longrightarrow\quad
x\in\sema{\varphi}{X}\,.
\end{displaymath}
\end{theorem}

We showed above that simulations for functors in the class
\textbf{Order} reflected and preserved all kinds of properties. Instead, now we
can only prove one implication, that corresponding to the reflection of
properties. This is so because down-natural $\nu$-orders have a natural
direction.

Exactly in the same way as we did with down-natural $\nu$-orders, we
can define the corresponding class of up-natural $\nu$-orders:

\begin{definition}
Let $F:\Sets\lra \Sets$ be a functor, $AP$ a set of atomic
propositions and $\nu: F\Rightarrow\ps(AP)$ a natural
transformation. We say that $\<$ is an up-natural $\nu$-order if $u\<
u'$ implies $\nu(u)\subseteq\nu(u')$.
\end{definition}

As we did for down-natural $\nu$-orders, we define a subclass
of \textbf{Order}:

\begin{definition}
The class \bfup{Up-Natural $\nu$-Order} is the subclass of \bfup{Order}
where all orders are up-natural.
\end{definition}

\begin{theorem}
Let $R$ be a simulation between coalgebras $c:X\lra FX$ and
$d:Y\lra FY$ on the same polynomial functor $F$ in the class
\bfup{Up-Natural $\nu$-Order}, and let $\varphi$ be a temporal formula.
Then, for all $x\in X$ and $y\in Y$ such that $xR y$:
\begin{displaymath}
x\in\sema{\varphi}{X}\quad \Longrightarrow\quad
y\in\sema{\varphi}{Y}\,.
\end{displaymath}
\end{theorem}

\section{Conclusions}

The main goal of this paper was to study under what assumptions
coalgebraic simulations reflect properties. In our way towards the
proof of this result, we were also able to prove reflection and
preservation of properties by coalgebraic bisimulations. For
expressing the properties we used Jacobs' temporal logic
\cite{Jacobs07}, later extended with atomic propositions using the
idea presented in \cite{Kurz-thesis}.

That coalgebraic bisimulations reflect and preserve properties
expressed in modal logic is a well-known topic (e.g,
\cite{Jacobs00,Kurz-thesis,Pattinson-thesis}), but not so the
corresponding results for simulations. The main difficulty is that
Hughes and Jacobs' notion of simulation is defined by means of an
arbitrary functorial order which bestows them with a high degree
of freedom. We have dealt with this by restricting the class of
functorial orders (although even so we are not able of obtaining a
general result) and by restricting also the class of allowed
functors.

In order to get more general results on the subject, an
interesting path that we intend to explore is the search for a
canonical notion of simulation. This definition would provide us,
not only with a ``natural'' way to understand simulations but,
hopefully, would also give rise to ``natural" general results
about reflection of properties.

Another promising direction of research is the study of reflection
and preservation of properties in probabilistic systems, following
our results of \cite{dFPF07b} in combination with the ideas
presented in \cite{Hasuo06a,VR97,BSV04}.

\noindent
\textbf{Acknowledgement}

The authors would like to thank the anonymous referees for their comments and
suggestions.


\end{document}